\documentclass[journal]{IEEEtran}
\usepackage[utf8]{inputenc}
\usepackage{authblk}
\usepackage{graphicx}
\usepackage{textcomp}
\usepackage{stfloats}
\usepackage{verbatim}
\usepackage{mathrsfs}
\usepackage{enumerate}
\usepackage{setspace}
\usepackage{amsmath}
\usepackage{amsthm}
\newtheorem{theorem}{Theorem}
\newtheorem{proposition}{Proposition}
\newtheorem*{notation}{Notation}
\newtheorem{corollary}{Corollary}
\usepackage{xcolor}
\newtheorem{remark}{Remark}
\newtheorem{definition}{Definition}
\newtheorem{assumption}{Assumption}
\newtheorem{lemma}{Lemma}
\usepackage{amssymb}
\usepackage{mathtools}
\usepackage{mathabx}
\usepackage{algorithmic}
\usepackage{array}
\usepackage{mdwmath}
\usepackage{mdwtab}
\usepackage{eqparbox}
\usepackage{url}
\usepackage{epstopdf}
\definecolor{darkcandyapplered}{rgb}{0.64, 0.0, 0.0}
\definecolor{frenchblue}{rgb}{0.0, 0.45, 0.73}
\definecolor{burntorange}{rgb}{0.8, 0.33, 0.0}
\definecolor{fluorescentorange}{rgb}{1.0, 0.75, 0.0}
\definecolor{green(ryb)}{rgb}{0.4, 0.69, 0.2}
 \usepackage[labelfont=bf, skip =6pt, font = small]{caption}
\usepackage{balance}
\usepackage{blindtext}
\usepackage{hyperref}
\usepackage{listings}
\begin{document}
\title{Adaptive Control of Heterogeneous Platoons with Guaranteed Collision Avoidance
\\
\thanks{This work is supported by “IIT Palakkad Technology IHub Foundation Doctoral
Fellowship IPTIF/HRD/DF/027/SEP38”.\par 
Ashutosh Chandra Pandey and Sayan Basu Roy are with the Department
of Electronics and Communication Engineering, Indraprastha
Institute of Information Technology Delhi, New Delhi, India (emails : ashutoshpa@iiitd.ac.in, sayan@iiitd.ac.in). \par
Simone Baldi is with the School of Mathematics, Southeast University, Nanjing, China (email : s.baldi@seu.edu.cn)}
}
  
\author{{Ashutosh Chandra Pandey}, {Sayan Basu Roy, \textit{Member, IEEE}} {and} {Simone Baldi, \textit{Senior Member, IEEE}}}

\maketitle

\begin{abstract}
This work proposes a framework for Cooperative Adaptive Cruise Control of a vehicular platoon characterized by unidirectional communication and heterogeneous parameters. In the proposed framework, the actual (heterogeneous) platoon is made to converge to a reference (homogeneous) platoon via adaptive laws designed using of set-theoretic model reference adaptive control.
Yet, in contrast to the state-of-art that is based on ensuring collision avoidance on the reference platoon dynamics only, the approach we propose can ensure collision avoidance on the actual platoon dynamics. This result is possible thanks to the introduction of a novel concept of virtual platoon, only used for analysis, but that does not interact with the actual platoon.
The stability and convergence properties of the proposed framework are established using Lyapunov-based analysis in conjunction with the aforementioned virtual platoon concept.
\end{abstract}
\begin{IEEEkeywords}
virtual platoon, collision avoidance, external positivity, cooperative adaptive cruise control, heterogeneous vehicular platoon, string stability.
\end{IEEEkeywords}
\IEEEoverridecommandlockouts
\IEEEpeerreviewmaketitle
\section{Introduction} \label{Introduction}
In recent years, the continuous advances in vehicle-to-vehicle communication have been opening the field of connected automated vehicles, paving the way to cooperative automated driving. Several studies suggest that cooperative automated vehicles can lead to substantial increase in traffic capacity due to the deployment of vehicular platoons \cite{jia2015survey,shladover2005automated}. Nowadays, many vehicles are already equipped with adaptive cruise control (ACC) technology, enabling to maintain a desired inter-vehicle distance \cite{783618,guo2012autonomous,4220653} with the help of on-board sensors like lasers and radars. Cooperative Adaptive Cruise Control (CACC) extends the ACC technology by adding wireless vehicle-to-vehicle communication, thus enhancing the connectivity of vehicles beyond on-board sensing \cite{6747309,10202202}. As compared to ACC, CACC ensures a smaller inter-vehicle gap \cite{ploeg2013lp}, and provides better performance in terms of string stability \cite{abou2017adaptive} and disturbance decoupling \cite{wijnbergen2020existence} criteria. String stability describes the attenuation of disturbances \cite{ploeg2013lp,8914474,1438390} (e.g., induced by sudden brake or acceleration), as they propagate downstream along the string of vehicles composing the platoon.\\ 
The most standard CACC protocols consider predecessor-follower unidirectional \cite{abou2017adaptive} and sometimes bidirectional \cite{kwon2014adaptive,baldi2020establishing} communication among adjacent vehicles. Typical desirable properties of the platoon are asymptotic stability \cite{swaroop1996string,703548} and string stability \cite{van2006impact,ploeg2013lp,wang2015string}. However, it has been shown that asymptotic stability and string stability cannot automatically ensure collision avoidance within adjacent vehicles in the platoon. External positivity has been regarded as a promising property to ensure collision avoidance within the platoon \cite{lunze2018adaptive,schwab2021design}. External positivity guarantees a positive output (i.e., inter-vehicle distance) for a positive input (i.e., the velocity of the predecessor vehicle). Using external positivity, collision avoidance is studied in \cite{lunze2018adaptive,schwab2021design} for ACC, and later extended in \cite{10130071} for CACC.\par
One of the pioneering works in CACC is \cite{ploeg2013lp} which assumes a known homogeneous actual platoon (AP). If the AP is heterogeneous, then homogeneous dynamics can be used to design a homogeneous reference platoon (RP) in the framework of model reference adaptive control (MRAC) \cite{abou2017adaptive,baldi2020establishing}, so as to deal with the situation that not only the AP is heterogeneous, but its parameters are possibly unknown.
The main idea behind the MRAC strategy is that desirable properties like asymptotic stability, string stability and external positivity can be imposed on the RP. Then, the AP can asymptotically attain the desirable properties by converging to the dynamics of the RP. Although effective, this strategy has the drawback that the desirable properties can only be attained asymptotically, which implies that external positivity (and thus collision avoidance) for the RP does not necessarily imply collision avoidance on the AP at all times. This represents a major open problem in the literature that is addressed in this work by introducing a new platoon concept and new analysis tools.\par 
A first significant contribution of this work is the introduction of a virtual platoon (VP) concept that, differently from the RP, is not part of the control design, but is only used for analysis. The advantage of the VP is that, because it does not interact with the AP as it happens with the RP (refer to Fig. \ref{CACC_Platoon_Structure}), it can be used to ensure properties of string stability and collision avoidance for the entire platoon, differently from the standard analysis in the literature based on a predecessor-follower pair \cite{abou2017adaptive, baldi2020establishing}. Our analysis shows that both RP and AP converge to VP asymptotically; hence, external positivity and string stability properties are established for the entire AP asymptotically. Then, a second significant contribution of this work is to prove that not only the AP asymptotically achieves external positivity, but that it actually attains collision avoidance at all times. This marks a difference with ensuring collision avoidance for the RP only. This result is possible thanks to the introduction of a set-theoretic MRAC architecture \cite{arabi2018set}, in place of a standard MRAC architecture, that allows to use a generalized barrier Lyapunov function \cite{tee2009barrier}. A novel Lyapunov analysis is presented to guarantee that not only the tracking errors between the AP and RP converge asymptotically to zero, but also that a worst-case user-defined bound can be guaranteed for the tracking error. This bound in turn can be used to ensure collision avoidance at all times not only for the RP but also for the AP. In this strategy, the VP model acts as a robust analysis tool and facilitates the appropriate use of Barbalat's Lemma \cite{Khalil:1173048}. With the aforementioned tracking error bound and the VP's external positivity condition, we derive a lower bound on the standstill distance of the AP ensuring collision avoidance at all times. \par
The paper is organized as follows. Preliminary information is provided in Section \ref{Preliminaries}. Section \ref{Platoon Dynamics} provides the dynamics for the CACC platooning and the baseline control law. Section \ref{Virtual Platoon Dynamics} presents the dynamics of the VP and analysis of asymptotic stability, external positivity, and string stability. In Section \ref{Adaptive CACC Augmentation}, RP dynamics, adaptive augmentation, and platoon convergence are given. Section \ref{Collision Avoidance} provides the condition for collision avoidance. 
\begin{notation}
   $\mathbb{R}$, $\mathbb{R}^{n }$ and $\mathbb{R}^{n \times m}$ represent the sets of real numbers, real vectors of dimension $n$ and real matrices of dimension ${n \times m}$, respectively. $\mathbb{R}_+$ denotes set of positive real numbers. $\mathbb{N}$ is set of natural numbers. We use $\left | \cdot \right |$ to denote the absolute value of a constant and $\left\|\cdot\right\|$ for the $2$ norm of a vector and the Frobenius norm of a matrix. The notation $\sup\left | \cdot\right |$ is the least upper bound of a function. $\mathfrak{L}^{-1}$ denotes the inverse Laplace transform. $\mathcal{L}_{2}$ and $\mathcal{L}_{\infty}$ denote the space of square integrable signals and the space of bounded signals, respectively. $I_{n \times n}$ is the identity matrix of dimension ${n \times n}$. The zero and the one vector of dimension $n$ are indicated respectively by $0_n = [0 \; 0 \; \cdots \; (n \; \text{times}) \; \cdots \; 0]^T$ and ${1}_n = [1 \; 1 \; \cdots \; (n \; \text{times}) \; \cdots \; 1]^T$, while the zero matrix of dimension $n \times m$ is represented by $0_{n \times m} = [0_n \; 0_n \; \cdots \; (m \; \text{times}) \; \cdots \; 0_n]$. The symbol ``$ \otimes$'' is used for Kronecker product.
\end{notation}
\section{Preliminaries}\label{Preliminaries}

In this section, some key concepts of generalized barrier Lyapunov function \cite{arabi2018set}, string stability \cite{ploeg2013lp,abou2017adaptive}, and external positivity \cite{lunze2018adaptive,10130071} are recalled.\\ 
\begin{definition}\label{blf}
   (Generalized barrier Lyapunov function \cite{arabi2018set}) Let $\left \| \eta \right \|_M = \sqrt{\eta^T M \eta}$ be the weighted Euclidean norm, with $\eta \in \mathbb{R}^p$, and $M \in \mathbb{R}^{p \times p}$ being symmetric positive definite. Consider a function $\psi (\cdot): \Omega_c \rightarrow \mathbb{R}$, where 
    \begin{equation}
        \Omega_c = \left \{\left \| \eta \right \|_M\in \mathbb{R}_+\ | \ \left \| \eta \right \|_M<c\right \},
    \end{equation}
    with $c\in \mathbb{R}_+$ a user-defined constant bound. The function $\psi(\cdot)$ is called a generalized barrier Lyapunov function if the following conditions are satisfied:
    \begin{enumerate}[i)]
        \item If $\left \| \eta \right \|_M = 0$, then $\psi (\left \| \eta \right \|_M) = 0$.
        \item $\psi (\left \| \eta \right \|_M) > 0$ inside $\Omega_c$ excluding $\left \| \eta \right \|_M = 0$.
        \item $\lim_{\left \| \eta \right \|_M \rightarrow c} \psi (\left \| \eta \right \|_M) = \infty$.
        \item $\psi (\left \| \eta \right \|_M)$ is continuously differentiable on $\Omega_c$.
        \item $\psi^{'} (\left \| \eta \right \|_M) > 0$ inside $\Omega_c$, where
        \begin{equation}\label{psi derivative wrt nx}
            \psi^{'} (\left \| \eta \right \|_M) = \frac{d\psi (\left \| \eta \right \|_M)}{d\left \| \eta \right \|_M^2}.
        \end{equation}
        \item Inside $\Omega_c$, the following holds
        \begin{equation}
            2 \psi^{'} (\left \| \eta \right \|_M) \left \| \eta \right \|_M^2 - \psi (\left \| \eta \right \|_M) > 0.
        \end{equation}
    \end{enumerate}
\end{definition}
The term generalized barrier Lyapunov function comes from the fact that the conditions in Definition \ref{blf} generalize the notion of barrier Lyapunov function \cite{tee2009barrier}. A possible generalized barrier Lyapunov function, also used in this work, is
\begin{equation}\label{blf candidate equation}
            \psi (\left \| \eta \right \|_{M}) = \frac{\left \| \eta \right \|_{M}^2}{c -\left \| \eta \right \|_{M}}
        \end{equation}
with $c$ being user-defined. The above function satisfies all the properties in Definition \ref{blf}. \par
\begin{definition}\label{External Positivity}
    (External Positivity \cite{lunze2018adaptive,10130071}) Consider the linear system
    \begin{equation}
    \begin{split}
        &\dot \chi (t) = \mathcal{A} \chi (t) + \mathcal{B}  \mu (t)\\
        & \Upsilon (t)=\mathcal{C} \chi (t) +\mathcal{D} \mu (t)
    \end{split}
    \end{equation}
    with zero initial conditions. Let $G(s) = \mathcal{C}(sI - \mathcal{A})^{-1}\mathcal{B}+\mathcal{D}$ be the transfer function of the system. The system is said to be externally positive (i.e., $ \mu(t)\geq 0, \forall t \geq 0 \Rightarrow \Upsilon(t) \geq 0, \forall t \geq 0$), if and only if its impulse response is non-negative, i.e., 
    \begin{equation}
        g(t)= \mathfrak{L}^{-1}\left \{  G(s) \right \} \geq 0, \qquad \forall t \geq 0.
    \end{equation}
\end{definition}
\begin{definition}\label{string stability}
     (String Stability \cite{ploeg2013lp,abou2017adaptive}) 
     Consider the cascaded interconnected system
    \begin{equation}
    \begin{split}
        &\dot { \chi}_j (t) = \mathcal{A}_j \chi_j (t) + \mathcal{B}_j \Upsilon_{j-1} (t)\\
        & \Upsilon_j (t)=\mathcal{C}_j \chi_j (t)
    \end{split}
    \end{equation}
    with zero initial conditions, $1 \leq j\leq N$, and $\Upsilon_0(t)$ an external input. Let $G_j(s) = \mathcal{C}_j(sI - \mathcal{A}_j)^{-1}\mathcal{B}_j$ be the transfer function from $\Upsilon_{j-1}(\jmath \omega)$ to $\Upsilon_j(\jmath \omega)$. The cascaded interconnected system is said to be string stable if
\begin{equation}\label{String Stability eq}
    \sup_{\omega}|G_j(\jmath  \omega)| = \sup_{\omega} \left|\frac{\Upsilon_j(\jmath \omega)}{\Upsilon_{j-1}(\jmath \omega)}\right| \leq 1, \; \forall \omega, \; \forall j.
\end{equation}
where $\jmath$ is the complex unit.
\end{definition}
\section{Platoon Model}\label{Platoon Dynamics}
Consider a heterogeneous platoon consisting of $N+1$ vehicles as shown in Fig. \ref{unidirectional_platoon}, where $v_j (t)$ is the velocity $\text{(m/s)}$ of vehicle $j$ and $s_{j-1,j}(t)$ is the distance $\text{(m)}$ between vehicle $j$ and its predecessor vehicle $j-1$. Consider the set $\alpha_N = \left \{ j\in \mathbb{N}\mid  1 \leq j\leq N \right \}$ with $j=0$ used for the leading vehicle. The objective of the platoon is to maintain a desired distance $s_{d,j}(t)$  between vehicle $j$ and its predecessor. A constant time headway (CTH) spacing policy \cite{rajamani2002semi} is used to define the desired distance $\text{(m)}$
\begin{equation}\label{CTH}
 s_{d,j}(t) = r_j+h v_j(t)  ,\qquad j \in \alpha_N, \\
\end{equation}
where $h>0$ is the time gap (s) and $r_j>0$ is the standstill distance (m) for vehicle $j$. The distance between vehicle $j$ and vehicle $j-1$ is calculated as
\begin{equation}\label{actual spacing between the vehicle}
s_{j-1,j}(t) = y_{i-1}(t)-y_j(t)-L_j, \qquad j \in \alpha_N, 
\end{equation}
where $y_j(t)\in \mathbb{R}$ is the rear bumper position (m) and $L_j>0$ is length (m) of vehicle $j$ as in Fig. \ref{unidirectional_platoon}. 
Then, the spacing error for vehicle $j$ can be defined as
\begin{equation}\label{error dynamics}
    \begin{split}
         e_j(t)&= s_{j-1,j}(t) - s_{d,j}(t)\\ 
     &=(y_{j-1}(t)-y_j(t)-L_j)- (r_j+h v_j(t)).
    \end{split}
\end{equation}
 In the following, we may omit the time dependence whenever obvious. The dynamics of vehicle $j$ is defined as \cite{abou2017adaptive}
    \begin{gather}\label{platoon dynamics}
        \begin{matrix}
  \begin{pmatrix}
            \dot e_{j}\\
            \dot v_j\\
            \dot a_j
        \end{pmatrix}
        =
        \begin{pmatrix}
            v_{j-1} - v_j - h a_j\\
            a_j \\
            -\frac{1}{\tau_j} a_j +\frac{1}{\tau_j} \Lambda_j u_j 
        \end{pmatrix}, &\quad j \in \alpha _N ,
\end{matrix}
\end{gather}
    where $u_j(t) \in \mathbb{R}$ and $a_j (t) \in \mathbb{R}$ are the input (${\text{m}}/{\text{s}^2}$) and the acceleration (${\text{m}}/{\text{s}^2}$), $ \Lambda_j>0$ is the engine efficiency and $\tau_j>0$ is the engine time constant (s) for vehicle $j$ \cite{abou2017adaptive}. In this work, we take $\Lambda_j$ and $\tau_j$ to be unknown, thus requiring an adaptive design.\\
\begin{figure}[t!]
    \centering
    \includegraphics[height=4cm,width=9cm]{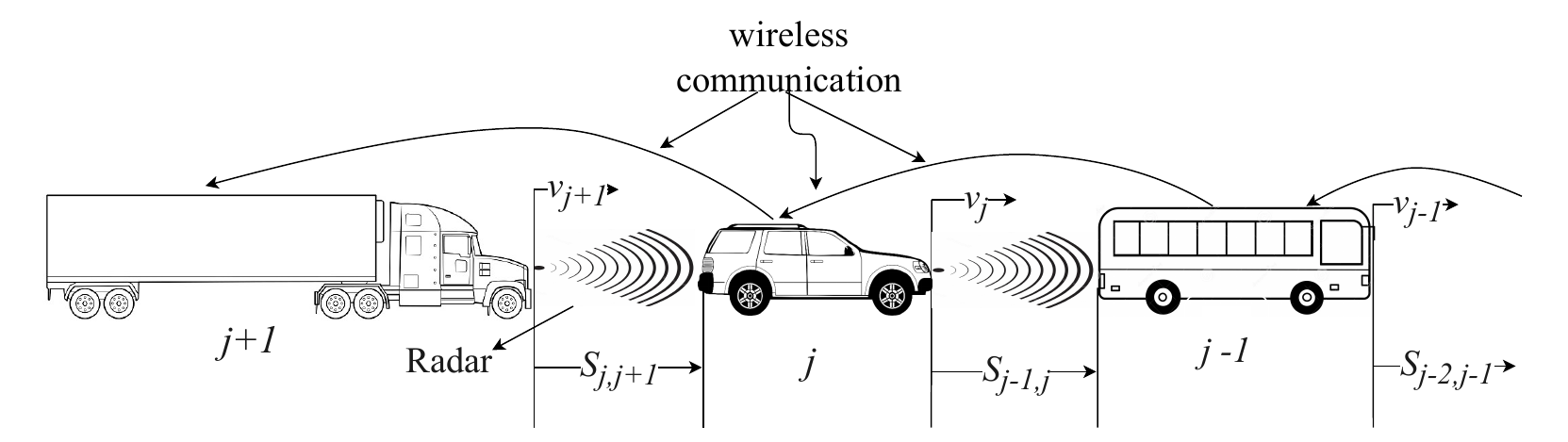}
        \caption{ CACC heterogeneous vehicle platoon equipped with unidirectional communication \cite{abou2017adaptive}}
    \label{unidirectional_platoon}
\end{figure} 
The interconnection of the dynamics \eqref{platoon dynamics} will be referred to as AP. The control objective is to design $u_j$ in \eqref{platoon dynamics} such that
\begin{itemize}
    \item All the spacing errors in the AP converges to zero, i.e., $\lim_{t\rightarrow \infty} e_j(t) = 0$, $\forall j \in \alpha _N $,
    \item collision avoidance is achieved in the AP, i.e., $s_{j-1,j}(t) > 0,\ \forall t\geq 0,\ \forall j\in\alpha_{N}$.
\end{itemize}
Other properties guaranteed by the proposed design (e.g., string stability and external positivity) will be presented later by making use of the VP concept.\\
The control input $u_j(t)$ is taken in the following form \cite{Robustbook}
    \begin{equation}\label{augmented control input}
        u_j (t) = u_{bl,j}(t) + u_{ad,j}(t), \quad \forall j \in \alpha_N \cup \left\{0 \right\}
    \end{equation}
where $u_{bl,j}(t)$ and $u_{ad,j}(t)$ denote the baseline controller and the adaptive controller, respectively. Hereafter, we introduce the baseline controller, while the adaptive controller is designed in Section \ref{Adaptive CACC Augmentation}.
\subsection{Baseline Controller}
 A baseline controller has been proposed in the literature \cite{ploeg2013lp} under the assumption that all vehicles are homogeneous (i.e., $\Lambda_j=1$ and $\tau_j$ is identical $\forall j \in \alpha_N$ in \eqref{platoon dynamics}). Such a controller takes the form
    \begin{equation}\label{baseline control dynamics}
        h \dot u_{bl,j}= -u_{bl,j} + k_p e_j + k_d \dot e_j + u_{bl,j-1},
    \end{equation}    
where $k_p>0$ and $k_d>0$ are the controller design parameters. The feedback from $u_{bl,j-1}$ is obtained via wireless communication, in line with the standard unidirectional communication in a CACC protocol (cf. Fig. \ref{unidirectional_platoon}). The leader vehicle control input is as follows
    \begin{equation}\label{leader control input}
        h \dot u_{bl,0} = -u_{bl,0} + u_{in},
    \end{equation}
where $u_{in}(t) \in\mathcal{L}_ \infty$ is the external input representing the desired acceleration $(\text{m}/\text{s}^2 )$ of the leader vehicle. Since the vehicles dynamic in \eqref{platoon dynamics} are taken heterogeneous, the assumption used in \cite{ploeg2013lp} fails, so that the baseline controller must be augmented by an adaptive term as in \eqref{augmented control input}. 
\begin{remark}
    To facilitate external positivity and string stability analysis, a homogeneous VP is now introduced. It is worth stressing that the VP is only used for analysis and does not interact with the RP or the AP as shown in Fig. \ref{CACC_Platoon_Structure}. This marks a difference with existing literature where the RP is used for analysis as well as for interacting with the AP \cite{abou2017adaptive,10130071}. 
\end{remark}

\section{Virtual Platoon Dynamics and Its Properties}\label{Virtual Platoon Dynamics}
A homogeneous VP is defined by using nominal homogeneous parameters  in the dynamics \eqref{platoon dynamics}- $\tau_j=\bar{\tau}>0$ and $\Lambda_j = 1, \; \forall j \in \alpha_N$ and using the baseline control dynamics in \eqref{baseline control dynamics}. The resulting dynamics is
\begin{gather}\label{new reference dynamics detail}
    \begin{split}
    \underset{\dot x_{v,j}}{\underbrace{
\begin{bmatrix}
\dot{e}_{v,j}\\ 
\dot{v}_{v,j}\\ 
\dot{a}_{v,j}\\ 
\dot{u}_{bl,v,j}
\end{bmatrix}}}
&=
\underset{A_m}{\underbrace{
\begin{bmatrix}
 0 & -1 & -h & 0\\
            0 &  0 & 1 & 0\\
            0 &  0 & -\frac{1}{\Bar{\tau}} & \frac{1}{\Bar{\tau}}\\
            \frac{k_p}{h} & -\frac{k_d}{h} & -k_d & -\frac{1}{h}
\end{bmatrix}}}
\underset{ x_{v,j}}{\underbrace{
\begin{bmatrix}
e_{v,j}\\ 
v_{v,j}\\ 
a_{v,j}\\ 
u_{bl,v,j}
\end{bmatrix}}}\\
&+
\underset{B_w}{\underbrace{
\begin{bmatrix}
            1& 0\\
            0 & 0\\
            0 & 0\\
            \frac{k_d}{h} & \frac{1}{h}
\end{bmatrix}}}
\underset{w_{v,j-1}}{\underbrace{
\begin{bmatrix}
            v_{v,j-1}\\
            u_{bl,v,j-1}
\end{bmatrix}}}, \quad \forall j \in \alpha_N,
\end{split}
\end{gather}
which can be compactly represented as follows
    \begin{equation}\label{brief of new reference dynamics}
       \dot x_{v,j} = A_m x_{v,j} + B_w w_{v,j-1}, \quad \forall j\; \in \:  \alpha_N,
    \end{equation}
    where subscript $v$ denotes the VP, $x_{v,j}(t) \in \mathbb{R}^4$ is the state vector of vehicle $j$ in the VP, and $w_{v,j-1}(t) \in \mathbb{R}^2$ contains signals from the preceding vehicle.\\ 
Using \eqref{leader control input}, the dynamics of the leader of the VP can be written as follows
    \begin{gather}\label{vehicle dynamics for zero}
    \underset{\dot x_{v,0}}{\underbrace{
        \begin{bmatrix}
            \dot e_{v,0}\\ 
\dot v_{v,0}\\ 
\dot a_{v,0}\\ 
\dot u_{bl,v,0}
        \end{bmatrix}}}
        =
        \underset{A_{m,0}}{\underbrace{
        \begin{bmatrix}
            0 & 0 & 0 & 0\\
            0 &  0 & 1 & 0\\
            0 &  0 & -\frac{1}{\Bar{\tau}} & \frac{1}{\Bar{\tau}}\\
            0 & 0 & 0 & -\frac{1}{h}
        \end{bmatrix}}}
        \underset{x_{v,0}}{\underbrace{
        \begin{bmatrix}
            e_{v,0}\\ 
v_{v,0}\\ 
a_{v,0}\\
u_{bl,v,0}
        \end{bmatrix}}}
        +
        \underset{B_r}{\underbrace{
        \begin{bmatrix}
             0\\
            0\\
             0\\
           \frac{1}{h}
        \end{bmatrix}}}
       u_{in}.
    \end{gather}
     \begin{figure*}[ht!]
     \vspace{-0.5cm}
    \centering
      \includegraphics[height=8cm,width=18cm]{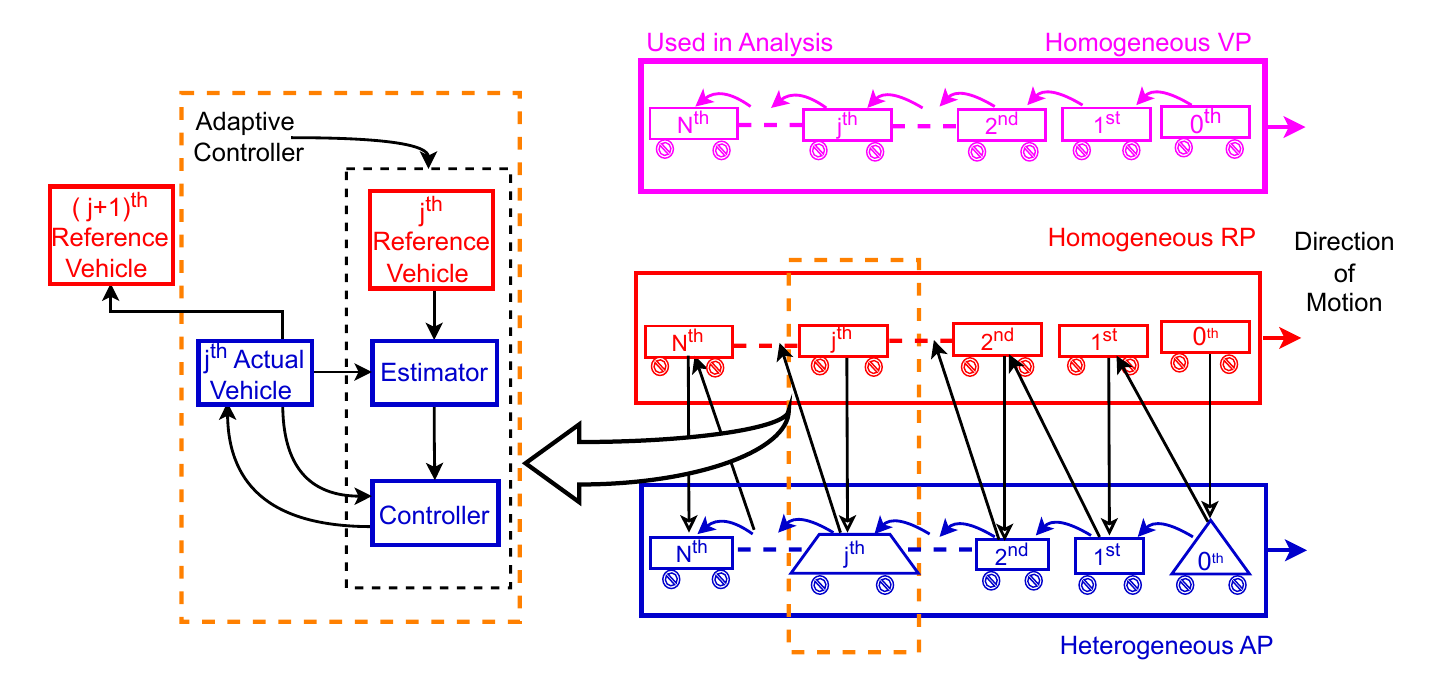}
        \caption{CACC platoon structure: actual platoon (AP), reference platoon (RP) and virtual platoon (VP).}
    \label{CACC_Platoon_Structure}
\end{figure*}
    \begin{remark}
         Since the leader has no predecessor, $e_0$ in \eqref{vehicle dynamics for zero} is a fictitious zero state. It is only introduced to ensure consistency with the other vehicle dynamics. Without loss of generality, a reduced-dimensional state can be defined as $\underline{x}_{v,0} =(v_{v,0} \quad a_{v,0} \quad u_{bl,v,0})$ which will be used subsequently to perform stability analysis for the VP, and will also be defined for the RP and the AP.
    \end{remark}
    \subsection{Stability Analysis of the Virtual Platoon}\label{String Stability Analysis}
According to Definition \ref{string stability}, string stability can be used to describe the ability to attenuate external disturbances as they propagate through the platoon due to the interconnected vehicle dynamics. To analyze the string stability of the VP, let us write the interconnected system dynamics in a compact form, resulting from \eqref{brief of new reference dynamics}-\eqref{vehicle dynamics for zero}
\begin{gather}\label{interconnected form for new reference dynamics}
\begin{split}
\underset{\dot {\bar{x}}_v}{\underbrace{
    \begin{bmatrix}
\dot{x}_{v,0}\\ 
\dot{x}_{v,1}\\ 
\vdots \\ 
\dot{x}_{v,N}
\end{bmatrix}}}
&=
\underset{\bar{A}_m}{\underbrace{
\begin{bmatrix}
A_{m,0} &  &  & 0_{4\times 4}\\ 
\bar{B}_w & A_m &  & 0_{4\times 4}\\ 
 &  \ddots & \ddots  & \\ 
0_{4\times 4}&  & \bar{B}_w & A_m 
\end{bmatrix}}}
\underset{\bar{x}_v}{\underbrace{
    \begin{bmatrix}
{x}_{v,0}\\ 
{x}_{v,1}\\ 
\vdots \\ 
{x}_{v,N}
\end{bmatrix}}}\\ &
+
\underset{\bar{B}}{\underbrace{
\begin{bmatrix}
B_r\\ 
0_{4\times 1}\\ 
\vdots \\ 
0_{4\times 1}
\end{bmatrix}}}
u_{in},
\end{split}
\end{gather}
where $\bar{B}_w = [0_{4 \times 1} \quad B_w(:,1) \quad 0_{4 \times 1} \quad B_w(:,2)] \in \mathbb{R}^{4 \times 4}$, $\bar{x}_v \in \mathbb{R}^{4(N+1) \times 1}$, $\bar{A}_m \in \mathbb{R}^{4(N+1) \times 4(N+1)}$ and $\Bar{B} \in \mathbb{R}^{4(N+1) \times 1}$.
The equilibrium point
    \begin{equation}\label{equilibrium point}
    \begin{split}
       & \bar {x}_{eq,v} =  1_{N+1} \otimes \underbrace{(0 \quad v_c \quad 0 \quad 0)^T}_{\triangleq x_{eq,v}}  \quad,\\  & \textrm{for} \quad x_{v,0}(0) =x_{eq,v} \quad \textrm{and} \quad u_{in} = 0 .
    \end{split}
    \end{equation}
of \eqref{interconnected form for new reference dynamics} has been proven to be asymptotically stable in \cite{ploeg2013lp}, where $v_c$ is any constant velocity, provided that the following Routh–Hurwitz conditions are satisfied
    \begin{equation}\label{Routh–Hurwitz conditions}
        h > 0, \quad k_p >0, \quad k_d>0,\quad k_d>\Bar{\tau} k_{p}.
    \end{equation}
    Let us define the constant equilibrium velocity in such a way that when $\lim_{t\to\infty}u_{in}(t)=0$, which implies $\lim_{t\to\infty}a_{v,0}(t)=0$, then 
\begin{equation}\label{eq:v_c_new}
    v_{c,new}=\lim_{t\to\infty} v_{v,0} (t)
\end{equation}
exists and finite.
\begin{theorem}\label{Theorem about VP spacing error}
 The vector $\bar{x}_v (t)$ of the VP defined in \eqref{interconnected form for new reference dynamics} converges to the equilibrium point $\bar {x}_{eq,v,new}=1_{N+1} \otimes (0 \quad v_{c,new} \quad 0 \quad 0)^T $, where $v_{c,new}$ is defined in \eqref{eq:v_c_new}.  
\end{theorem}
\begin{proof}
 Since $u_{in}(t)$ is such that $x_{v,0}(t)$ converges to $(0 \quad v_{c,new} \quad 0 \quad 0)^T$, 
 following the arguments of \cite{ploeg2013lp}, it can be shown that the origin of the dynamics of $(\bar{x}_v(t)-\bar{x}_{eq,v})$ is asymptotically converging to zero since $\lim_{t\to\infty} u_{in}(t)=0$.
\end{proof}
  \subsection{External Positivity and String Stability Analysis of the Virtual Platoon}\label{External Positivity Analysis}
  Let us define the pseudo spacing between vehicle $(j)$ and vehicle $(j-1)$ in the VP
\begin{equation}\label{spacing}
     \begin{split}
     \bar{s}_{v,j-1,j}(t) &= y_{v,j-1}(t) - y_{v,j}(t) - L_j- r_j\\
         &= s_{v,j-1,j}(t)-r_j.
    \end{split}
\end{equation}
Let us introduce the following assumption:
  \begin{assumption}\label{assumption for VP}
      The initial conditions satisfy
      \begin{equation}
      \begin{split}
          &\bar{s}_{v,j-1,j} (0) \geq 0, \; i.e., \; s_{v,j-1,j} (0)\geq r_j\\
         &  v_{v,j}(0) \geq 0, \; \forall j \in \alpha_N.
      \end{split}
      \end{equation}
  \end{assumption}
  This assumption guarantees that the initial condition of the platoon is physically well posed, e.g., there is no collision at $t=0$ and velocities are non-negative.\\
 The external positivity and the string stability for the VP can be proven along similar lines as the literature \cite{lunze2018adaptive,10130071}, via the following two lemmas.
\begin{lemma}\label{External positivity of inter-vehicle spacing}
    (External positivity of inter-vehicle spacing) Consider the VP dynamics \eqref{brief of new reference dynamics} and \eqref{vehicle dynamics for zero}. The linear system defined by the transfer function
    $G_{v,j}(s) = \frac{v_{v,j}(s)}{v_{v,j-1}(s)}$ is externally positive. Furthermore, the linear system defined by the transfer function $F_{v,j}(s) = \frac{\bar{s}_{v,j-1,j} (s)}{v_{v,j-1}(s)}$ is also externally positive.
\end{lemma}
\begin{proof}
From the VP dynamics \eqref{brief of new reference dynamics}, $G_{v,j}(s)$ is calculated as \cite{abou2017adaptive}
\begin{equation}\label{controlled system transfer function for velocity output}
    G_{v,j}(s)=\frac{\bar{\tau} s^3 + s^2 + k_d s + k_p}{(hs+1)(\bar{\tau} s^3 + s^2 + k_d s + k_p)}=\frac{1}{hs+1}.
\end{equation}
The time derivative of $\bar{s}_{v,j-1,j}(t)$ in \eqref{spacing} gives
\begin{equation}\label{distance derivative dynamics}
    \dot{\bar{s}}_{v,j-1,j}(t) = v_{v,j-1}(t)-v_{v,j}(t).
\end{equation}
From \eqref{controlled system transfer function for velocity output} and \eqref{distance derivative dynamics}, the transfer function $F_{v,j}(s)$ is calculated as
\begin{equation}\label{controlled system transfer function for spacing output}
    F_{v,j}(s)=\frac{h}{hs+1}.
\end{equation}
The impulse response of \eqref{controlled system transfer function for velocity output} is $g_{v,j}(t) = h^{-1}e^{-h^{-1}t}, \forall t\geq0$. We have $g_{v,j}(t) \geq 0, \forall t\geq0$, which satisfies Definition \ref{External Positivity}. Furthermore, the impulse response of \eqref{controlled system transfer function for spacing output} is $f_{v,j}(t) = e^{-h^{-1}t}, \forall t\geq0$. We have $f_{v,j}(t) \geq 0, \forall t\geq0$, which also satisfies Definition \ref{External Positivity}.\\
    \end{proof}
\begin{remark} \label{Collision avoidance of reference platoon}
From Definition \ref{External Positivity} and Lemma \ref{External positivity of inter-vehicle spacing}, it is concluded that a positive input $v_{v,j-1}(t)$ would ensure that both the outputs $v_{v,j}(t)$ and $\bar{s}_{v,j-1,j}(t)$, $\forall t\geq 0$ remain positive, so that the collision avoidance of VP is ensured under Assumption \ref{assumption for VP}.
\end{remark}
\begin{remark}\label{remark for assumption of initial condition}
      The goal of this work is to guarantee collision avoidance in the AP, which is different from the state-of-the-art approach of guaranteeing collision avoidance in the RP \cite{10130071}. Such a goal requires a different analysis and the introduction of $\bar{s}_{v,j-1,j}(t)$ instead of $s_{v,j-1,j}(t)$, as it will become evident in Section \ref{Collision Avoidance}.
  \end{remark}
\begin{lemma}\label{lemma for string stability}
    The VP in \eqref{brief of new reference dynamics} is string stable.
\end{lemma}
\begin{proof}
    String stability of the linear system defined by the transfer function $G_{v,j}(s) = \frac{a_{v,j}(s)}{a_{v,j-1}(s)}$ (equivalently, $G_{v,j}(s) = \frac{v_{v,j}(s)}{v_{v,j-1}(s)}$) can be analyzed from \eqref{String Stability eq} and \eqref{controlled system transfer function for velocity output},
     and the following is concluded
    \begin{equation}\label{string stability for w}
        |G_{v,j}(\jmath  \omega)| = \frac{1}{\sqrt{h^2 \omega^2 +1}}
    \end{equation}
    implying, $|G_{v,j}(\jmath  \omega)| \leq 1,\; \forall \omega, \; \forall j\; \in \:  \alpha_N \: \setminus \: \left \{ 0\right \}$.  
Therefore, from Definition \ref{string stability}, the VP in \eqref{brief of new reference dynamics} is string stable.
\end{proof}
The above properties of VP are going to be useful in the analysis of subsequently designed MRAC architecture. 
\section{MRAC Augmented CACC Platooning}\label{Adaptive CACC Augmentation}
In order to address heterogeneity in the platoon, the adaptive controller in \eqref{augmented control input} is now designed. First, we design homogeneous RP dynamics. Differently from the VP dynamics that do not interact with the AP, the RP dynamics are part of the adaptive control design and thus interact with the AP. 
\subsection{Reference Dynamics}\label{Reference Dynamics}
 The homogeneous RP dynamics  is designed as
\begin{gather}\label{reference dynamics detail}
    \begin{split}
    \underset{\dot x_{r,j}}{\underbrace{
\begin{bmatrix}
\dot{e}_{r,j}\\ 
\dot{v}_{r,j}\\ 
\dot{a}_{r,j}\\ 
\dot{u}_{bl,r,j}
\end{bmatrix}}}
&=
\underset{A_m}{\underbrace{
\begin{bmatrix}
 0 & -1 & -h & 0\\
            0 &  0 & 1 & 0\\
            0 &  0 & -\frac{1}{\Bar{\tau}} & \frac{1}{\Bar{\tau}}\\
            \frac{k_p}{h} & -\frac{k_d}{h} & -k_d & -\frac{1}{h}
\end{bmatrix}}}
\underset{ x_{r,j}}{\underbrace{
\begin{bmatrix}
e_{r,j}\\ 
v_{r,j}\\ 
a_{r,j}\\ 
u_{bl,r,j}
\end{bmatrix}}}\\
&+
\underset{B_w}{\underbrace{
\begin{bmatrix}
            1& 0\\
            0 & 0\\
            0 & 0\\
            \frac{k_d}{h} & \frac{1}{h}
\end{bmatrix}}}
\underset{w_{j-1}}{\underbrace{
\begin{bmatrix}
            v_{i-1}\\
            u_{bl,i-1}
\end{bmatrix}}},\quad \forall j \in \alpha_N,
\end{split}
\end{gather}  
subject to the initial condition
    \begin{equation}\label{assumption for RP}
        x_{r,j}(0) = x_{v,j} (0), \; \forall j \in \alpha_N
    \end{equation}
where subscript $r$ denotes the RP, $x_{r,j}(t) \in \mathbb{R}^4$ and $w_{j-1}(t) \in \mathbb{R}^2$ are the RP vector and the AP input vector, respectively. The dynamics of the leader of the RP can be written as follows

 \begin{gather}\label{vehicle dynamics for zero for RP}
    \underset{\dot {\underline{x}}_{r,0}}{\underbrace{
        \begin{bmatrix}
\dot v_{r,0}\\ 
\dot a_{r,0}\\ 
\dot u_{bl,r,0}
        \end{bmatrix}}}
        =
        \underset{\underline{A}_{m,0}}{\underbrace{
        \begin{bmatrix}
              0 & 1 & 0\\
               0 & -\frac{1}{\Bar{\tau}} & \frac{1}{\Bar{\tau}}\\
              0 & 0 & -\frac{1}{h}
        \end{bmatrix}}}
        \underset{\underline{x}_{r,0}}{\underbrace{
        \begin{bmatrix}     
v_{r,0}\\ 
a_{r,0}\\
u_{bl,r,0}
        \end{bmatrix}}}
        +
        \underset{\underline{B}_r}{\underbrace{
        \begin{bmatrix} 
            0\\
             0\\
           \frac{1}{h}
        \end{bmatrix}}}
       u_{in}.
    \end{gather}
\begin{remark}
    Due to the coupling between the AP and the RP, it is not straightforward to derive string stability and external positivity for the whole RP in the sense of Definitions \ref{External Positivity} and \ref{string stability}. On the other hand, such properties can be directly shown for VP (cf. Lemma \ref{External positivity of inter-vehicle spacing} and Lemma \ref{lemma for string stability}). In subsequent analysis, it is shown that the RP and the AP ultimately converge to the VP. Unlike existing analysis tools that rely on properties defined for pairs of adjacent predecessor-follower vehicles \cite{10130071}, the analysis we present does not rely on a unidirectional predecessor-follower interconnection and is potentially applicable to a broader class of platoon interconnections \cite{10387736}.
\end{remark}
In order to write the interconnected form for the RP dynamics in \eqref{reference dynamics detail}, after some manipulation, the following dynamics is achieved
\begin{equation}\label{manipulated brief reference dynamics}
       \dot x_{r,j} = A_m x_{r,j} + B_w w_{r,j-1} + B_w \tilde{w}_{j-1}, \quad \forall j\; \in \:  \alpha_N,
    \end{equation}
    where $\tilde{w}_{j-1} (t)= w_{j-1}(t) - w_{r,j-1}(t)$.
    Therefore, the interconnected form for the dynamics in \eqref{reference dynamics detail} is as follows
\begin{gather}\label{interconnected form for reference dynamics}
    \begin{split}
\underset{ \dot {\bar{x}}_r}{\underbrace{
    \begin{bmatrix}
\dot{x}_{r,0}\\ 
\dot{x}_{r,1}\\ 
\vdots \\ 
\dot{x}_{r,N}
\end{bmatrix}}}
&\hspace{-0.05cm}=\hspace{-0.05cm}
\underset{\bar{A}_m}{\underbrace{
\begin{bmatrix}
A_{m,0} &  &  & 0_{4\times 4}\\ 
\bar{B}_w & A_m &  & 0_{4\times 4}\\ 
 &  \ddots & \ddots  & \\ 
0_{4\times 4}&  & \bar{B}_w & A_m 
\end{bmatrix}}}
\underset{\bar{x}_r}{\underbrace{
    \begin{bmatrix}
{x}_{r,0}\\ 
{x}_{r,1}\\ 
\vdots \\ 
{x}_{r,N}
\end{bmatrix}}}
\\
&+
\underset{\bar{B}}{\underbrace{
\begin{bmatrix}
B_r\\ 
0_{4\times 1}\\ 
\vdots \\ 
0_{4\times 1}
\end{bmatrix}}}
u_{in}\hspace{-0.05cm}+\hspace{-0.05cm}
\underset{{B_{\tilde {\bar {w}}}}}{\underbrace{
\begin{bmatrix}
0_{4\times 2}   &  & 0_{4\times 2}\\ 
 B_w &  & 0_{4\times 2}\\ 
   \ddots & \ddots & \\ 
0_{4\times 2}  & 0_{4\times 2} & B_w
\end{bmatrix}}}
\underset{{\tilde {\bar {w}}}}{\underbrace{
\begin{bmatrix}
\tilde{w}_0\\ 
\tilde{w}_1\\ 
\vdots \\ 
\tilde{w}_{N-1}
\end{bmatrix}}},
\end{split}
\end{gather}
where $x_{r,0}(t)=x_{v,0}(t) \; \forall t\geq 0$, $\bar{x}_r (t)\in \mathbb{R}^{4(N+1) \times 1}$, $\tilde {\bar {w}}(t) \in \mathbb{R}^{2N \times 1}$ and $B_{\tilde {\bar {w}}} \in \mathbb{R}^{4(N+1) \times 2N}$.
\subsection{MRAC Design}\label{MRAC}
 Substituting \eqref{augmented control input} into \eqref{platoon dynamics} results in
\begin{gather}\label{Uncertain vehicle model detail}
\begin{split}
\underset{\dot x_j}{\underbrace{
    \begin{bmatrix}
\dot{e}_j\\ 
\dot{v}_j\\ 
\dot{a}_j\\ 
\dot{u}_{bl,j}
\end{bmatrix}}}
&=
\underset{A_j}{\underbrace{
\begin{bmatrix}
0 & -1 & -h & 0\\ 
0 & 0 & 1 & 0\\ 
0 & 0 & -\frac{1}{\tau_j} & \frac{\Lambda_j}{\tau_j}\\ 
\frac{k_p}{h} & -\frac{k_d}{h} & -k_d & -\frac{1}{h}
\end{bmatrix}}}
\underset{x_j}{\underbrace{
\begin{bmatrix}
e_j\\ 
v_j\\ 
a_j\\ 
u_{bl,j}
\end{bmatrix}}}\\
&+
\underset{B_w}{\underbrace{
\begin{bmatrix}
            1& 0\\
            0 & 0\\
            0 & 0\\
            \frac{k_d}{h} & \frac{1}{h}
\end{bmatrix}}}
\underset{w_{j-1}}{\underbrace{
\begin{bmatrix}
v_{j-1}\\
            u_{bl,j-1}
\end{bmatrix}}}
+
\underset{B_u}{\underbrace{
\begin{bmatrix}
0\\ 
0\\ 
1\\ 
0
\end{bmatrix}}}
\underset{\theta_j}{\underbrace{
\left (\frac{\Lambda _j}{\tau _j}  \right )}}u_{ad,j},
\end{split}
 \end{gather}
which can be compactly written as
     \begin{equation}\label{Uncertain vehicle model} 
    \dot x_j = A_jx_j + B_w w_{j-1} + B_u \theta_j u_{ad,j}, \quad \forall j \in \alpha_N ,
\end{equation}
where $x_{j}(t)\in \mathbb{R}^{4}$, $A_j\in \mathbb{R}^{4\times 4}$, $B_{u}\in \mathbb{R}^{4}$ and $\theta_{j}\in \mathbb{R}$.

The following ideal controller gain would make the AP dynamics in \eqref{Uncertain vehicle model} match the RP dynamics in \eqref{reference dynamics detail}.
\begin{lemma}\label{Lemma for Kxj}
    There exist a constant vector $k_{x,j} \in \mathbb {R}^4$ such that
\begin{equation}\label{matching condition}
\begin{split}
     &A_m = A_j + B_u \theta_j k_{x,j}^T \quad \forall j \in \alpha_N\\
\end{split}
\end{equation}
\end{lemma}
\begin{proof}
    It can be shown that the following expression of $k_{x,j}^T$ satisfies \eqref{matching condition}
    \begin{equation*}
        k_{x,j}^T = \left [ 0 \quad 0 \quad \frac{(\Bar{\tau} - \tau_j)}{\Lambda_j \Bar{\tau}} \quad \frac{\tau_j - \Lambda_j \Bar{\tau}}{\Lambda_j \Bar{\tau}} \right ], \quad \forall j \in \alpha_N.
    \end{equation*}
\end{proof} 
It should be noted that the ideal controller gain $k_{x,j}$ in Lemma \ref{Lemma for Kxj} requires knowledge of $\Lambda_j$ and $\tau_j$, which are unknown and therefore, $u_{ad,j}(t)$ is defined as follows
\begin{equation}\label{adaptive control input}
    u_{ad,j}(t) = \hat{k}_{x,j}^T(t) x_j(t),  \quad \forall j \in \alpha_N,
\end{equation}
where $\hat{k}_{x,j}(t) \in \mathbb{R}^{4}$ is the online estimate of $k_{x,j}$. 
Define the tracking error as
\begin{equation}\label{tracking error}
    \Tilde{x}_j(t) = x_j(t) - x_{r,j}(t),  \quad \forall j \in \alpha_N.
\end{equation}
Define the parameter estimation error as
\begin{equation}\label{parameter estimation error}
    \Tilde{k}_{x,j}(t) = \hat{k}_{x,j}(t) -k_{x,j},  \quad \forall j \in \alpha_N.
\end{equation}
Using \eqref{Uncertain vehicle model}, \eqref{matching condition}, \eqref{adaptive control input}, \eqref{tracking error} and \eqref{parameter estimation error}, the tracking error dynamics is
\begin{equation}\label{state tracking error dynamics}
    \dot {\Tilde{x}}_j = A_m \Tilde{x}_j + B_u \theta_j \Tilde{k}_{x,j}^T x_j,  \quad \forall j \in \alpha_N.
\end{equation}
Let the parameter update law for the gain in \eqref{adaptive control input} be designed as
\begin{equation}\label{control law}
\dot {\hat{k}}_{x,j} = - \Gamma_{x,j} \psi^{'} (\left \| \Tilde{x}_j \right \|_{P_m}) x_j \Tilde{x}_j^T P_m B_u, \quad \forall j \in \alpha_N,
\end{equation}
where $P_m$ is a positive definite matrix satisfying
\begin{equation}\label{PD equation}
     A_m^T P_m + P_m A_m=-Q_m, 
\end{equation}
for any design matrix $Q_m = Q_m^T >0$.
 From \eqref{control law}, the parameter estimation error dynamics is as follows
\begin{equation} \label{ control law adaptation, no composite adaptation}
    \dot {\Tilde{k}}_{x,j} = - \Gamma_{x,j} \psi^{'} (\left \| \Tilde{x}_j \right \|_{P_m}) x_j \Tilde{x}_j^T P_m B_u, \quad \forall j \in \alpha_N.
\end{equation}
\begin{theorem}\label{lyapunov theorem}
     For the dynamics of AP in \eqref{Uncertain vehicle model} and the dynamics of RP in \eqref{reference dynamics detail}, the control input \eqref{adaptive control input} and the adaptive update law \eqref{control law} ensure the following:
     \begin{enumerate}
         \item The origin of the error $(\tilde{x}_j(t),\tilde{k}_{x,j}(t))$ dynamics in \eqref{state tracking error dynamics} and \eqref{ control law adaptation, no composite adaptation}, is Lyapunov stable $\forall j \in \alpha_N \cup \left\{0 \right\}$.
         \item The state tracking error is guaranteed to satisfy
     \begin{equation}\label{error bound}
         \left \| \tilde{x}_j (t) \right \|_{P_m} < c,\quad \forall t \geq 0,\quad \forall j \in \alpha_N.
     \end{equation}
     where $c$ is defined in \eqref{blf candidate equation}.
     \end{enumerate}
\end{theorem}
\begin{proof}
    $1)$ Let us define the following radially unbounded Lyapunov candidate
\begin{equation}\label{Lyapunov candidate adaptive}
\begin{split}
    V_j(\tilde{x}_j (t),\tilde{k}_{x,j} (t)) &= \frac{1}{2} \psi (\left \| \tilde{x}_j \right \|_{P_m})\\
    & + \frac{\theta_j}{2}\tilde{k}_{x,j}^T (t) \Gamma_{x,j}^{-1}\tilde{k}_{x,j} (t), \quad \forall j \in \alpha_N.
\end{split}
\end{equation}
Taking the time derivative of \eqref{Lyapunov candidate adaptive}, substituting \eqref{state tracking error dynamics} in it, and using \eqref{PD equation}, the following is obtained
\begin{equation}\label{time derivative of Lyapunov candidate adaptive}
\begin{split}
    \dot V_j(\tilde{x}_j (t),\tilde{k}_{x,j}(t)) =& -\frac{1}{2}\psi^{'} (\left \| \Tilde{x}_j \right \|_{P_m}) \tilde{x}_j^T Q_m \tilde{x}_j\\  &+\theta_j \psi^{'} (\left \| \Tilde{x}_j \right \|_{P_m}) x_{j}^{T}\tilde{k}_{x,j}B_u^T P_m \tilde{x}_{j} \\
    &+\theta_j \dot{\hat {k}}_{x,j}^T \Gamma_{x,j}^{-1}\tilde{k}_{x,j}.
    \end{split}
\end{equation}
Considering the adaptive law as in \eqref{control law}, $\dot V_j(\tilde{x}_j(t),\tilde{k}_{x,j}(t))$ in \eqref{time derivative of Lyapunov candidate adaptive} reduces to
\begin{equation}\label{NSD time derivative adaptive}
    \dot V_j(\tilde{x}_j(t),\tilde{k}_{x,j}(t)) \leq -\frac{1}{2}\psi^{'} (\left \| \Tilde{x}_j \right \|_{P_m}) \tilde{x}_j^T Q_m \tilde{x}_j \leq 0.
\end{equation}
In \eqref{NSD time derivative adaptive}, $\dot V_j(\tilde{x}_j(t),\tilde{k}_{x,j}(t))$ is negative semi definite, therefore the origin of $(\tilde{x}_j(t)(t),\tilde{k}_{x,j}(t))$ is Lyapunov stable. 
$2)$ From the fact that $\dot V_j(\tilde{x}_j (t),\tilde{k}_{x,j} (t)) \leq 0$, we have $V_j(\tilde{x}_j (t),\tilde{k}_{x,j} (t)) \leq V_j(\tilde{x}_j (0),\tilde{k}_{x,j} (0)) < \infty$. Hence, $( \psi (\left \| \tilde{x}_j(t) \right \|_{P_m}),\tilde{k}_{x,j}(t)) \in \mathcal{L}_{\infty}$ from \eqref{Lyapunov candidate adaptive}, and consequently, $\left \| \tilde{x}_j \right \|_{P_m} < c, \forall t \geq 0$ from the Definition \ref{blf}, which implies $\tilde{x}_j (t) \in \mathcal{L}_\infty $.  
\end{proof}
\subsection{MRAC Design for Leader}\label{MRAC for Leader}
From \eqref{augmented control input} and \eqref{leader control input}, the dynamics of the leader of the AP can be written as follows
\begin{gather}\label{vehicle dynamics for zero for AP}
\begin{split}
    \underset{\dot {\underline{x}}_{0}}{\underbrace{
        \begin{bmatrix}
\dot v_{0}\\ 
\dot a_{0}\\ 
\dot u_{bl,0}
        \end{bmatrix}}}
        &=
        \underset{\underline{A}_0}{\underbrace{
        \begin{bmatrix}
             0 & 1 & 0\\
             0 & -\frac{1}{\tau_0} & \frac{\Lambda_{0}}{\tau_0}\\
             0 & 0 & -\frac{1}{h}
        \end{bmatrix}}}
        \underset{\underline{x}_{0}}{\underbrace{
        \begin{bmatrix}
v_{0}\\ 
a_{0}\\
u_{bl,0}
        \end{bmatrix}}}
        \\&+
\underset{\underline{B}_u}{\underbrace{
\begin{bmatrix}
0\\ 
1\\ 
0
\end{bmatrix}}}
\underset{\theta_0}{\underbrace{
\left (\frac{\Lambda _0}{\tau _0}  \right )}}u_{ad,0}
        +
        \underset{\underline{B}_r}{\underbrace{
        \begin{bmatrix}
            0\\
             0\\
           \frac{1}{h}
        \end{bmatrix}}}
       u_{in}
       \end{split}
    \end{gather}
    Let us assume a matrix $\underline{A}_c$ as
    \begin{gather}
        \underline{A}_c = \begin{bmatrix}
0 & 1 & 0 \\
\underline{a}_{c, \tilde{v}} & \underline{a}_{c,\tilde{a}} & \underline{a}_{c,\tilde{u}_{bl}} \\
0 & 0 & -\frac{1}{h} \\
\end{bmatrix}
    \end{gather}
    where $\underline{A}_c$ is Hurwitz for the following Routh-Hurwitz condition:
    \begin{equation}
        \underline{a}_{c, \tilde{v}} <0, \; \underline{a}_{c,\tilde{a}} <0, \; \underline{a}_{c,\tilde{u}_{bl}} \in \mathbb{R}.
    \end{equation}
    The following ideal controller gain would make the AP dynamics in \eqref{vehicle dynamics for zero for AP} match the RP dynamics in \eqref{vehicle dynamics for zero for RP}.
\begin{lemma}
    There exist constant vectors $k_{\underline{x}_{0}} \in \mathbb {R}^3$ and $k_{ \tilde{\underline{x}}_{0}} \in \mathbb {R}^3$ such that
\begin{equation}\label{matching condition of leader AP}
\begin{split}
     & \underline{A}_{m,0} = \underline{A}_0 + \underline{B}_u \theta_0 k_{\underline{x}_{0}}^T,\\
     & \underline{A}_c = \underline{A}_{m,0} + \underline{B}_u \theta_0 k_{ \tilde{\underline{x}}_{0}}^T.
\end{split}
\end{equation}
\end{lemma}
\begin{proof}
    It can be shown that the following $k_{\underline{x}_{0}}$ and $k_{ \tilde{\underline{x}}_{0}}$ satisfy \eqref{matching condition of leader AP}
    \begin{equation*}
        k_{\underline{x}_{0}}^T = \left [0 \quad \frac{(\Bar{\tau} - \tau_0)}{\Lambda_0 \Bar{\tau}} \quad \frac{\tau_0 - \Lambda_0 \Bar{\tau}}{\Lambda_0 \Bar{\tau}} \right ]
    \end{equation*}
    \begin{equation*}
        k_{ \tilde{\underline{x}}_{0}}^T=\left[\frac{\tau_0\underline{a}_{c, \tilde{v}}}{\Lambda_0} \quad \frac{(\bar{\tau}\underline{a}_{c,\tilde{a}}+1)\tau_0}{\Lambda_0 \bar{\tau}} \quad \frac{(\bar{\tau}\underline{a}_{c,\tilde{u}_{bl}}-1)\tau_0}{\Lambda_0 \bar{\tau}}\right].
    \end{equation*}
\end{proof} 
It should be noted that the ideal controller gains $k_{\underline{x}_{0}}$ and $k_{ \tilde{\underline{x}}_{0}}$ require knowledge of $\Lambda_0$ and $\tau_0$, which are unknown and therefore, $u_{ad,0}(t)$ for \eqref{vehicle dynamics for zero for AP} is defined as follows
\begin{equation}\label{adaptive control input for zero}
    u_{ad,0}(t) = \hat{k}_{\underline{x}_{0}}^T(t) \underline{x}_{0}(t)+ \hat{k}_{ \tilde{\underline{x}}_{0}}^T(t) \tilde{\underline{x}}_{0}(t) ,  
\end{equation}
where $\hat{k}_{\underline{x}_{0}}(t) \in \mathbb{R}^{3}$ is the online estimate of $k_{\underline{x}_{0}}$ and $\hat{k}_{ \tilde{\underline{x}}_{0}}(t) \in \mathbb{R}^{3}$ is the online estimate of ${k}_{ \tilde{\underline{x}}_{0}}$, whereas $ \tilde{\underline{x}}_{0}(t)$ is the tracking error defined as
\begin{equation}\label{tracking error for leader}
    \tilde{\underline{x}}_{0}(t) = \underline{x}_{0}(t) - \underline{x}_{r,0}(t).
\end{equation}
Define parameter estimation error as
\begin{equation}\label{parameter estimation error for leader}
    \Tilde{k}_{\underline{x}_{0}}(t) = \hat{k}_{\underline{x}_{0}}(t) -k_{\underline{x}_{0}}.
\end{equation}
\begin{equation}\label{tilde parameter estimation error for leader}
    \Tilde{k}_{\tilde{\underline{x}}_{0}}(t) = \hat{k}_{\tilde{\underline{x}}_{0}}(t) -k_{\tilde{\underline{x}}_{0}}.
\end{equation}
Using \eqref{vehicle dynamics for zero for AP}, \eqref{matching condition of leader AP}, \eqref{adaptive control input for zero}, \eqref{tracking error for leader}, \eqref{parameter estimation error for leader} and \eqref{tilde parameter estimation error for leader}, the tracking error dynamics is
\begin{equation}\label{state tracking error dynamics for leader}
    \dot {\tilde{\underline{x}}}_{0} = \underline{A}_c \tilde{\underline{x}}_{0} + \underline{B}_u \theta_0 \Tilde{k}_{\underline{x}_{0}}^T \underline{x}_{0}+\underline{B}_u \theta_0\Tilde{k}_{\tilde{\underline{x}}_{0}}^T \tilde{\underline{x}}_{0} 
\end{equation}
Let the parameter update law for the gain in \eqref{adaptive control input for zero} be designed as
\begin{equation}\label{control law for leader}
\dot {\hat{k}}_{\underline{x}_{0}} = - \Gamma_{\underline{x}_{0}} \psi^{'} (\left \| \tilde{\underline{x}}_{0} \right \|_{P_{m,0}}) \underline{x}_{0} \tilde{\underline{x}}_{0}^T P_{m,0} \underline{B}_u,
\end{equation}
\begin{equation}\label{tilde control law for leader}
\dot {\hat{k}}_{\tilde{\underline{x}}_{0}} = - \Gamma_{\underline{x}_{0}} \psi^{'} (\left \| \tilde{\underline{x}}_{0} \right \|_{P_{m,0}}) \tilde{\underline{x}}_{0} \tilde{\underline{x}}_{0}^T P_{m,0} \underline{B}_u,
\end{equation}
where $P_{m,0}$ is a positive definite matrix satisfying
\begin{equation}\label{PD equation for leader}
     A_{c}^T P_{m,0} + P_{m,0} A_{c}=-Q_{c}, 
\end{equation}
for any design matrix $Q_{c} = Q_{c}^T >0$.
From \eqref{control law for leader} and \eqref{tilde control law for leader}, the parameter estimation error dynamics are as follows
\begin{equation} \label{control law adaptation, no composite adaptation for leader}
    \dot {\tilde{k}}_{\underline{x}_{0}} = - \Gamma_{\underline{x}_{0}} \psi^{'} (\left \| \tilde{\underline{x}}_{0} \right \|_{P_{m,0}}) \underline{x}_{0} \tilde{\underline{x}}_{0}^T P_{m,0} \underline{B}_u,
\end{equation}
\begin{equation}\label{tilde control law adaptation}
\dot {\tilde{k}}_{\tilde{\underline{x}}_{0}} = - \Gamma_{\tilde{\underline{x}}_{0}} \psi^{'} (\left \| \tilde{\underline{x}}_{0} \right \|_{P_{m,0}}) \tilde{\underline{x}}_{0} \tilde{\underline{x}}_{0}^T P_{m,0} \underline{B}_u.
\end{equation}
\begin{theorem}\label{lyapunov theorem leader}
     For the leader dynamics of AP in \eqref{vehicle dynamics for zero for AP} and the leader dynamics of RP in \eqref{vehicle dynamics for zero for RP}, the control input \eqref{adaptive control input for zero} and the adaptive update laws \eqref{control law for leader}, \eqref{tilde control law for leader} ensure the following:
     \begin{enumerate}
         \item The origin of the error $(\tilde{\underline{x}}_{0}(t), \Tilde{k}_{\underline{x}_{0}} (t), \Tilde{k}_{\tilde{\underline{x}}_{0}}(t))$ dynamics in \eqref{state tracking error dynamics for leader},\eqref{control law adaptation, no composite adaptation for leader} and \eqref{tilde control law adaptation}, is Lyapunov stable.
         \item The state tracking error is guaranteed to satisfy 
     \begin{equation}\label{error bound leader}
         \left \| \tilde{\underline{x}}_{0}(t) \right \|_{P_{m,0}} < c,\quad \forall t \geq 0,\quad \forall j \in \alpha_N.
     \end{equation}
     where $c$ is defined in \eqref{blf candidate equation}.
     \end{enumerate}
\end{theorem}
\begin{proof}
$1)$ Let us define the following radially unbounded Lyapunov candidate
\begin{equation}\label{Lyapunov candidate adaptive leader}
\begin{split}
    &V_0(\tilde{\underline{x}}_{0}(t), \Tilde{k}_{\underline{x}_{0}} (t), \Tilde{k}_{\tilde{\underline{x}}_{0}}(t)) = \frac{1}{2} \psi (\left \| \tilde{\underline{x}}_{0}(t) \right \|_{P_{m,0}}) \\
    &+ \frac{\theta_0}{2}\tilde{k}_{\underline{x}_{0}}^T (t) \Gamma_{\underline{x}_{0}}^{-1}\tilde{k}_{\underline{x}_{0}} (t)+ \frac{\theta_0}{2}\tilde{k}_{\tilde{\underline{x}}_{0}}^T (t) \Gamma_{\tilde{\underline{x}}_{0}}^{-1}\tilde{k}_{\tilde{\underline{x}}_{0}} (t).
\end{split}
\end{equation}
Taking the time derivative of \eqref{Lyapunov candidate adaptive leader}, substituting \eqref{state tracking error dynamics for leader} in it, and using \eqref{PD equation for leader}, the following can be obtained
\begin{equation}\label{time derivative of Lyapunov candidate adaptive leader}
\begin{split}
    &\dot V_0(\tilde{\underline{x}}_{0}(t), \Tilde{k}_{\underline{x}_{0}} (t), \Tilde{k}_{\tilde{\underline{x}}_{0}}(t)) = -\frac{1}{2}\psi^{'} (\left \| \tilde{\underline{x}}_{0} \right \|_{P_{m,0}}) \tilde{\underline{x}}_{0}^T Q_c \tilde{\underline{x}}_{0}\\  &+\theta_0 \psi^{'} (\left \| \tilde{\underline{x}}_{0} \right \|_{P_{m,0}}) \underline{x}_{0}^{T}\tilde{k}_{\underline{x}_{0}}B_u^T P_{m,0} \tilde{\underline{x}}_{0}+\theta_0 \dot{\hat {k}}_{\underline{x}_{0}}^T \Gamma_{\underline{x}_{0}}^{-1}\tilde{k}_{\underline{x}_{0}} \\
    &+\theta_0 \psi^{'} (\left \| \tilde{\underline{x}}_{0} \right \|_{P_{m,0}}) \tilde{\underline{x}}_{0}^{T}\tilde{k}_{\tilde{\underline{x}}_{0}}B_u^T P_{m,0} \tilde{\underline{x}}_{0}+\theta_0 \dot{\hat {k}}_{\tilde{\underline{x}}_{0}}^T \Gamma_{\tilde{\underline{x}}_{0}}^{-1}\tilde{k}_{\tilde{\underline{x}}_{0}}.
    \end{split}
\end{equation}
Considering the adaptive law as in \eqref{control law for leader} and \eqref{tilde control law for leader}, $ \dot V_0(\tilde{\underline{x}}_{0}(t), \Tilde{k}_{\underline{x}_{0}} (t), \Tilde{k}_{\tilde{\underline{x}}_{0}}(t))$ in \eqref{time derivative of Lyapunov candidate adaptive leader} reduces to
\begin{equation}\label{NSD time derivative adaptive leader}
   \dot V_0(\tilde{\underline{x}}_{0}(t), \Tilde{k}_{\underline{x}_{0}} (t), \Tilde{k}_{\tilde{\underline{x}}_{0}}(t)) \leq -\frac{1}{2}\psi^{'} (\left \| \tilde{\underline{x}}_{0} \right \|_{P_{m,0}}) \tilde{\underline{x}}_{0}^T Q_c \tilde{\underline{x}}_{0} \leq 0.
\end{equation}
In \eqref{NSD time derivative adaptive leader}, $ \dot V_0(\tilde{\underline{x}}_{0}(t), \Tilde{k}_{\underline{x}_{0}} (t), \Tilde{k}_{\tilde{\underline{x}}_{0}}(t))$ is negative semi definite, therefore the origin of $(\tilde{\underline{x}}_{0}(t), \Tilde{k}_{\underline{x}_{0}} (t), \Tilde{k}_{\tilde{\underline{x}}_{0}}(t))$ is Lyapunov stable. \\
$2)$ From the fact that $ \dot V_0(\tilde{\underline{x}}_{0}(t), \Tilde{k}_{\underline{x}_{0}} (t), \Tilde{k}_{\tilde{\underline{x}}_{0}}(t)) \leq 0$, we have $V_0(\tilde{\underline{x}}_{0}(t), \Tilde{k}_{\underline{x}_{0}} (t), \Tilde{k}_{\tilde{\underline{x}}_{0}}(t)) \leq V_0(\tilde{\underline{x}}_{0}(0), \Tilde{k}_{\underline{x}_{0}} (0), \Tilde{k}_{\tilde{\underline{x}}_{0}}(0)) < \infty$. Hence, $( \psi (\left \| \tilde{\underline{x}}_{0}(t) \right \|_{P_{m,0}}),\Tilde{k}_{\underline{x}_{0}} (t) \Tilde{k}_{\tilde{\underline{x}}_{0}}(t)) \in \mathcal{L}_{\infty}$  from \eqref{Lyapunov candidate adaptive leader}, and consequently, $\left \| \tilde{\underline{x}}_{0}(t) \right \|_{P_{m,0}} < c, \forall t \geq 0$ from the Definition \ref{blf}, which also implies $ \tilde{\underline{x}}_{0}(t) \in \mathcal{L}_\infty $. 
\end{proof}
\begin{remark}\label{corollary on tracking error of leader}
    Since $\tilde{e}_0(t)= e_0(t)-e_{0,r}(t)=0 \; \forall t\geq 0$, the state tracking error between the leader of the AP and the leader of the RP satisfies $\tilde{x}_0(t)= (\tilde{e}_0(t) \quad \tilde{v}_0(t) \quad \tilde{a}_0(t) \quad \tilde{u}_{bl,0}(t))\in \mathcal{L}_\infty$.
\end{remark}
\subsection{Platoon Convergence}
In this part of the work, we provide a proof of convergence of the AP dynamics in \eqref{Uncertain vehicle model} to the external positive and string stable VP dynamics as $t \rightarrow \infty$. Define the error between the RP dynamics in \eqref{interconnected form for reference dynamics} and the VP dynamics in \eqref{interconnected form for new reference dynamics} as
\begin{equation}\label{error dynamics between RD and GD}
    \bar{z}(t) = \bar{x}_v(t) - \bar{x}_r(t).
\end{equation}
 Calculating the time derivative of \eqref{error dynamics between RD and GD} and using the dynamics in \eqref{interconnected form for new reference dynamics} and \eqref{interconnected form for reference dynamics}, the following error dynamics is derived
\begin{equation}\label{derivative of error dynamics between RD and GD}
    \dot{\bar{z}}  = \Bar{A}_m \bar{z} - B_{\tilde {\bar {w}}} \tilde {\bar {w}}.
\end{equation}
\begin{theorem}\label{theorem on Z uub}
     The solution of \eqref{derivative of error dynamics between RD and GD} is uniformly ultimately bounded (UUB).
\end{theorem}
\begin{proof}
From the initial condition in \eqref{assumption for RP} and the error in \eqref{error dynamics between RD and GD} the following is achieved
\begin{equation}\label{initial condition of z}
    \bar{z}(0)= 0.
\end{equation}
We can further infer that, for $ \tilde {\bar {w}} (t) = 0$, the equilibrium point of \eqref{derivative of error dynamics between RD and GD} is $\bar {z}_{eq} =  1_{N+1} \otimes 0_{4} $. Furthermore, from the Routh–Hurwitz condition given in \eqref{Routh–Hurwitz conditions}, the system $\dot{\bar{z}} (t)  = \Bar{A}_m \bar{z}(t)$ is asymptotically stable around the equilibrium points.\\
Based on Theorem \ref{lyapunov theorem}, $\tilde{x}_j (t) \in \mathcal{L}_\infty$, therefore, $\tilde{w}_{j-1}(t) \in \mathcal{L}_\infty$ from the RP dynamics in \eqref{reference dynamics detail}. Furthermore, $\tilde {\bar {w}} (t) \in \mathcal{L}_\infty$ from the interconnected RP dynamics in \eqref{interconnected form for reference dynamics}. According to \eqref{derivative of error dynamics between RD and GD}, $\bar{z}$ is the state of an asymptotically stable linear system with bounded external input $\tilde {\bar {w}}(t)$. Therefore, any solution to \eqref{derivative of error dynamics between RD and GD} is UUB.
\end{proof}
\begin{theorem}\label{theorem on Platoon Convergence}
    The state tracking error in \eqref{tracking error} asymptotically converges to zero, i.e., $\lim_{t\rightarrow \infty } \left ( x_j(t) - x_{r,j}(t) \right ) = 0, \quad \forall j \in \alpha_N$. 
\end{theorem}
\begin{proof}
    Since $\tilde{k}_{x,j} (t) \in \mathcal{L}_\infty$ from Theorem \ref{lyapunov theorem} and $k_{x,j}$ is constant, thus $\hat{k}_{x,j} (t) \in \mathcal{L}_\infty$  from \eqref{parameter estimation error}. From \eqref{derivative of error dynamics between RD and GD}, $\bar{z}(t) \in \mathcal{L}_\infty$, and $\bar{x}_v (t) \in \mathcal{L}_\infty$ by design. Therefore, $\bar{x}_r (t) \in \mathcal{L}_\infty$ from the error dynamics \eqref{error dynamics between RD and GD}. Since $\bar{x}_r (t)$ is the vector of $x_{r,j} (t)$, then $x_{r,j} (t) \in \mathcal{L}_\infty$. Consequently, ${x}_j (t) \in \mathcal{L}_\infty$ and $u_{bl,j} (t) \in \mathcal{L}_\infty$, from the error dynamics in \eqref{tracking error}. From \eqref{adaptive control input}, $u_{ad,j} (t) \in \mathcal{L}_\infty$, hence $u_j (t) \in \mathcal{L}_\infty$ from \eqref{augmented control input}. In \eqref{state tracking error dynamics}, $A_m$ is Hurwitz; $B_u,\: \theta_j$ are constant; therefore, $\dot {\tilde{x}}_j (t) \in \mathcal{L}_\infty $. This concludes that $\tilde x (t)$ is uniformly continuous.\\
As from Definition \ref{blf} \cite{arabi2018set}, $\psi^{'}(\left \| \Tilde{x}_j(t) \right \|_M)$ is strictly increasing with respect to $(\left \| \Tilde{x}_j (t)\right \|_M)$, from \eqref{psi derivative wrt nx} and \eqref{blf candidate equation} the following holds
\begin{equation}\label{limit on psi das}
    \psi^{'}(\left \| \Tilde{x}_j \right \|_{P_m}) \geq \frac{1}{c} \quad \forall t \geq 0.
\end{equation}
Therefore, using \eqref{limit on psi das} in \eqref{NSD time derivative adaptive}, the following is achieved 
\begin{equation}\label{v dot in new form}
    \dot V_j(\tilde{x}_j,\tilde{k}_{x,j}) \leq -\frac{1}{c} \lambda_{min} (Q_m) \left \| \Tilde{x}_j \right \|^2\leq 0, 
\end{equation}
where $\lambda_{min} (Q_m)$ is the minimum eigenvalue of the positive definite matrix $Q_m$. Applying integration on both sides of \eqref{v dot in new form} w.r.t. time $t$ gives
\begin{equation}\label{integration of v dot}
    V(\infty) - V(0) \leq -\frac{1}{c} \lambda_{min} (Q_m) \int_{0}^{\infty} \left \| \Tilde{x}_j \right \|^2. 
\end{equation}
After rearranging \eqref{integration of v dot}, the following is concluded
\begin{equation}\label{rearranged integration of v dot}
    \frac{1}{c} \lambda_{min} (Q_m) \int_{0}^{\infty} \left \| \Tilde{x}_j \right \|^2 \leq  V(0) - V(\infty) < \infty.
\end{equation}
From \eqref{rearranged integration of v dot}, we have $\Tilde{x}_j (t) \in \mathcal{L}_2$. Therefore, from Theorem \ref{lyapunov theorem} and Remark \ref{corollary on tracking error of leader}, we have $\Tilde{x}_j (t) \in \mathcal{L}_2 \cap \mathcal{L}_ \infty \quad \forall j \in \alpha_N \cup \left\{0 \right\}$. Using Barbalat's Lemma (Lemma $8.2$
 of \cite{Khalil:1173048}), asymptotic convergence of $\tilde{x}_j (t)$ is ensured, i.e., $\lim_{t \rightarrow \infty}\tilde x_j (t) = 0 \quad \forall j \in \alpha_N \cup \left\{0 \right\}$.
\end{proof}
\begin{proposition}\label{theorem for convergence of v to r}
    The RP dynamics in \eqref{reference dynamics detail} asymptotically converges to the VP dynamics in \eqref{brief of new reference dynamics}, i.e., $\lim_{t\rightarrow \infty } \bar{z}(t) = 0$. Hence, the AP dynamics in \eqref{Uncertain vehicle model} also asymptotically converges to the VP dynamics, i.e., $\lim_{t\rightarrow \infty } \left ( x_{j}(t) - x_{v,j}(t) \right ) = 0$.
\end{proposition}
\begin{proof}
    Note that $\tilde {\bar {w}}(t)$ is a vector comprising $\tilde{w}_{j}(t)$'s and the elements $\tilde{w}_{j} (t)$ are included in the elements of $\tilde {x}_{j} (t),\quad \forall j \in \alpha_N \cup \left\{0 \right\}$. From Theorem \ref{theorem on Platoon Convergence}, $\lim_{t \rightarrow \infty}\tilde x_j (t) = 0 \quad \forall j \in \alpha_N \cup \left\{0 \right\}$. Thus, $\lim_{t \rightarrow \infty} \tilde {\bar {w}} (t) = 0 \quad \forall j \in \alpha_N \cup \left\{0 \right\}$.
Therefore, from \eqref{derivative of error dynamics between RD and GD}, the following is concluded
\begin{equation}
    \lim_{t\rightarrow \infty } \bar{z}(t) = 0.
\end{equation}
 It has been proven $\forall j \in \alpha_N$ that $\lim_{t\rightarrow \infty } \tilde{x}_{j}(t)  = 0$ in Theorem \ref{theorem on Platoon Convergence} and $\lim_{t\rightarrow \infty }  z_{j}(t)  = 0$. Therefore, the AP dynamics $x_{j}(t)$ is asymptotically tracks the VP dynamics $x_{v,j}(t)$, i.e., $\lim_{t\rightarrow \infty } \left ( x_{j}(t) - x_{v,j}(t) \right ) = 0, \quad \forall j \in \alpha_N \cup \left\{0 \right\}$.
\end{proof}
 From Theorem \ref{Theorem about VP spacing error} and Proposition \ref{theorem for convergence of v to r}, it can be shown that $\lim_{t \rightarrow \infty}v_0(t)= v_{c, new}$.
\begin{corollary}\label{corollary about AP spacing error}
    The spacing error of the AP converges to zero, i.e., $\lim_{t\rightarrow \infty }e_j(t) = 0$ $\forall j \in \alpha _N \cup\{0\} $.
\end{corollary}
\begin{proof}
    Since, from Theorem \ref{Theorem about VP spacing error}, $\lim_{t\rightarrow \infty }e_{v,j}(t) =0, \; \forall j \in \alpha_N \cup\{0\}$ and from Proposition \ref{theorem for convergence of v to r} $\lim_{t\rightarrow \infty } \left ( x_{j}(t) - x_{v,j}(t) \right ) = 0$, we have $\lim_{t\rightarrow \infty } \left ( e_{j}(t) - e_{v,j}(t) \right ) = 0$. It is concluded that $\lim_{t\rightarrow \infty }e_j(t) = 0$ $\forall j \in \alpha _N \cup\{0\}$.
\end{proof}

\begin{remark}
The importance of the Corollary \ref{corollary about AP spacing error} is to show that not only $\lim_{t\rightarrow \infty } {e}_{v,j}(t) = 0$, but also $\lim_{t\rightarrow \infty }{e}_{j}(t) = 0$, $\forall j\in\alpha_N \cup\{0\}$. 
\end{remark}
\section{Collision Avoidance in the Actual Platoon}\label{Collision Avoidance}
Lemma \ref{External positivity of inter-vehicle spacing} only guarantees collision avoidance of the VP, but it does not guarantee collision avoidance of the AP. Therefore, in this section, we derive the conditions for collision avoidance of the AP.\\
\begin{lemma}\label{lemma for error bound}
   The state tracking error $\tilde{x}_j^v(t) = x_j(t) - x_{v,j}(t)$ satisfies
    \begin{equation}\label{state tracking error bound between AP and VP}
        \left \| \tilde{x}_j^v(t) \right \| <\bar{\mathcal{Z}} + \Omega \sqrt{N}\bar{\mathcal{Z}},  \quad \forall j \in \alpha_N,
    \end{equation}
    where $\Omega =\max_{\forall t\in[0,\infty)} \int_{0}^{t}\left \| e^{\Bar{A}_m^{'}(t-\tau)} \right \|  \left \|B_{\tilde {\bar {w}}}^{'} \right \|d\tau$\footnote{Since $\Bar{A}_m^{'}$ is Hurwitz, $\Omega$ is a finite constant.} with $\Bar{A}_m^{'} = \Bar{A}_m (5:4(N+1),5:4(N+1))\in\mathbb{R}^{4N\times 4N}$, $B_{\tilde {\bar {w}}}^{'}=B_{\tilde {\bar {w}}}(5:4(N+1),:)\in\mathbb{R}^{4N\times 2N}$ 
        and $\bar{\mathcal{Z}} = \max (\frac{c}{\lambda_{min}(P_{m,0})},\frac{c}{\lambda_{min}(P_{m})})$, $\lambda_{min}(P_{m,0})$ and $\lambda_{min}(P_{m})$ being the minimum eigenvalues   of the matrices $P_{m,0}$ and $P_{m}$, respectively.
\end{lemma}
\begin{proof}
    Using the initial condition in \eqref{initial condition of z}, the dynamics in \eqref{derivative of error dynamics between RD and GD} is integrated on both sides with respect to $t$
\begin{equation}\label{integration of error dynamics between RD and GD}
   \bar{z}(t) = \int_{0}^{t}e^{\Bar{A}_m(t-\tau)}B_{\tilde {\bar {w}}}\tilde {\bar {w}}(\tau) d\tau.
\end{equation}
Using the bound on $\tilde{x}_j(t)$ in \eqref{error bound}, a bound for $\tilde{w}_{j}(t)$ is as follows
\begin{equation}\label{bound on error of external input}
    \left \| \tilde{w}_{j}(t) \right \|< \left \| \tilde{x}_{j}(t) \right \|<\bar{\mathcal{Z}}.
\end{equation}
Using \eqref{bound on error of external input}, bound for the $\tilde{\bar {w}}(t)$ is derived as follow
\begin{equation}\label{bound for error of compact w}
    \left \| \tilde {\bar {w}}(t) \right \|< \sqrt{N}\bar{\mathcal{Z}}.
\end{equation}
Using \eqref{integration of error dynamics between RD and GD} and \eqref{bound for error of compact w}, bound for $\bar{z}(t)$ is given as follow
\begin{equation}\label{final bound between compact reference models}
     \left \| \bar{z}(t) \right \|<\Omega \sqrt{N}\bar{\mathcal{Z}}.
\end{equation}
The following state tracking error bound can also be concluded from the bound in \eqref{final bound between compact reference models}
\begin{equation}\label{final bound for jth reference models}
   \left \| x_{v,j}(t) - x_{r,j}(t) \right \|= \left \| z_j(t) \right \|<\Omega \sqrt{N}\bar{\mathcal{Z}}.
\end{equation}
Using the bound for the $\left \|\tilde {x}_j(t) \right \|$ in \eqref{error bound} and $z_j(t)$ in \eqref{final bound for jth reference models}, the following inequality holds
\begin{equation}
\begin{split}
    \left \| \tilde{x}_j^v(t) \right \| &=  \left \| x_j(t) - x_{v,j}(t) \right \|\\ & \leq \left \| \tilde {x}_j(t)\right\| + \left \| z_j(t) \right \|\\
    &<\bar{\mathcal{Z}} + \Omega \sqrt{N}\bar{\mathcal{Z}}.
\end{split}
\end{equation}
This completes the proof.
\end{proof}
The following main result guarantees collision avoidance of the entire AP.
\begin{theorem}\label{theorem for Collision Avoidance}
    If
    \begin{equation}\label{inequality about standstill distance}
    \begin{split}
         r_j > (\bar{\mathcal{Z}} + \Omega \sqrt{N}\bar{\mathcal{Z}})+h(\bar{\mathcal{Z}} + \Omega \sqrt{N}\bar{\mathcal{Z}}),
    \end{split}
    \end{equation}
    the actual spacing between the vehicles in the AP is always positive, i.e., $ s_{j-1,j}(t) > 0$.
\end{theorem}
\begin{proof}
Let us define the spacing error of the VP for the spacing in \eqref{spacing}
\begin{equation}\label{spacing error of the VP}
    e_{v,j}(t) = s_{v,j-1,j}(t) - s_{d,v,j}(t)
\end{equation}
where $s_{d,v,j}(t) = r_j + h v_{v,j}(t)$. Further defining the error between \eqref{error dynamics} and \eqref{spacing error of the VP}, the following holds
\begin{equation}\label{error dynamics of the spacing error}
    \begin{split}
    \tilde{e}_j^v (t)& = e_j (t) - e_{v,j}(t)\\
     &= (\bar{s}_{j-1,j}(t)-\bar{s}_{v,j-1,j}(t)) -h \tilde{v}_j^v(t) 
    \end{split}
\end{equation}
where $\tilde{v}_j^v(t) =v_{j}(t)-v_{v,j}(t) $ and $\bar{s}_{j-1,j}(t) = s_{j-1,j}(t)-r_j$.\\
After rearranging \eqref{error dynamics of the spacing error}, the spacing of RP dynamics is written as follows
\begin{equation}\label{spacing for reference dynamics}
    \bar{s}_{v,j-1,j}(t)= \bar{s}_{j-1,j}(t)-\tilde{e}_j^v (t)-h \tilde{v}_j^v(t).
\end{equation}
Since from Lemma \ref{External positivity of inter-vehicle spacing} and Remark \ref{Collision avoidance of reference platoon}, $\bar{s}_{v,j-1,j}(t) \geq 0, \forall t \geq 0$, the following inequality is achieved from \eqref{spacing for reference dynamics}
\begin{equation}\label{spacing for uncertain platoon}
    \bar{s}_{j-1,j}(t)\geq \tilde{e}_j^v (t)+h \tilde{v}_j^v(t).
\end{equation}
The bound in \eqref{state tracking error bound between AP and VP} implies
\begin{equation}\label{bound on state error}
\begin{split}
    \left | \tilde{e}_j^v \right |<\bar{\mathcal{Z}} + \Omega \sqrt{N}\bar{\mathcal{Z}} \quad \left | \tilde{v}_j^v \right |<\bar{\mathcal{Z}} + \Omega \sqrt{N}\bar{\mathcal{Z}}.
\end{split}
\end{equation}
Using the bound in \eqref{bound on state error} in inequality \eqref{spacing for uncertain platoon}, the following bound is achieved
\begin{equation}\label{inequality bound for actual spacing}
\begin{split}
   \left |\tilde{e}_j^v (t)+h\tilde{v}_j^v(t)  \right | < (\bar{\mathcal{Z}} + \Omega \sqrt{N}\bar{\mathcal{Z}})+h(\bar{\mathcal{Z}} + \Omega \sqrt{N}\bar{\mathcal{Z}}).
\end{split}
\end{equation}
From \eqref{spacing for uncertain platoon} and \eqref{inequality bound for actual spacing}
\begin{equation}\label{bound for actual spacing}
\begin{split}
   \bar{s}_{j-1,j}(t) > -((\bar{\mathcal{Z}} + \Omega \sqrt{N}\bar{\mathcal{Z}})+h(\bar{\mathcal{Z}} + \Omega \sqrt{N}\bar{\mathcal{Z}})).
\end{split}
\end{equation}
Therefore, using this bound in \eqref{actual spacing between the vehicle}, the following inequality for the actual spacing is concluded
\begin{equation}\label{final bound for actual spacing}
\begin{split}
   s_{j-1,j}(t) > r_j -((\bar{\mathcal{Z}} + \Omega \sqrt{N}\bar{\mathcal{Z}})+h(\bar{\mathcal{Z}} + \Omega \sqrt{N}\bar{\mathcal{Z}})).
\end{split}
\end{equation}
From the inequality in \eqref{inequality about standstill distance} and \eqref{final bound for actual spacing}, the following is concluded for the actual spacing
\begin{equation}\label{inequality for s actual}
     s_{j-1,j}(t) > 0,\ \forall t\geq 0,\ \forall j\in\alpha_{N}.
\end{equation}
\end{proof}
\section{Conclusion}\label{Conclusion}
In contrast to analyzing predecessor-follower pairs as the state-of-the-art platooning methods, this work has proposed a virtual-platoon-based (VP-based) analysis, which offers a generalized framework ensuring
desirable properties (asymptotic stability, string stability, external positivity, etc.) for the entire platoon. A set-theoretic MRAC-based CACC architecture was designed for the actual platoon, subject to uncertain and possibly heterogeneous engine efficiency. Rigorous mathematical analysis, utilizing the Lyapunov technique in conjunction with the VP concept, demonstrated stability and the platoon's convergence of the state of the heterogeneous actual platoon to the state of a homogeneous reference platoon. Furthermore, the set-theoretic MRAC-based CACC architecture allows to satisfy a user-defined performance bound for the state tracking error. A critical limiting condition for the standstill distance is derived from the bound on the tracking error, ensuring that the actual platoon can avoid inter-vehicle collisions for all times including the initial transient interval. 
Exploring the proposed VP-base analysis in more general multi-hop communication protocols would be an exciting avenue for future research. 
\bibliographystyle{IEEEtran}
\bibliography{references}
\end{document}